\documentclass[a4paper,12pt]{article}
\usepackage{fullpage}
\usepackage{indentfirst}
\usepackage{verbatim}
\usepackage{amssymb,amsmath, amsthm}
\usepackage{slashed}
\usepackage{cite}
\usepackage[bookmarks]{hyperref}
\numberwithin{equation}{section}
\title{Linear and angular momentum spaces for Majorana spinors}
\author{Leonardo Pedro\\ Centro de Fisica Teorica de Particulas,
  Portugal\\ leonardo@cftp.ist.utl.pt}
\date{\today}
\theoremstyle{plain}
\newtheorem{thm}{Theorem}[section]
\newtheorem{lem}[thm]{Lemma}
\newtheorem{prop}[thm]{Proposition}
\newtheorem{rem}[thm]{Remark}
\newtheorem*{cor}{Corollary}
\theoremstyle{definition}
\newtheorem{defn}[thm]{Definition}

\theoremstyle{remark}

\begin{document}
\maketitle
\begin{abstract}
In a Majorana basis, the Dirac equation for a free spin one-half
particle is a 4x4 real matrix differential equation. The solution can
be a Majorana spinor, a 4x1 real column matrix, whose entries are real
functions of the space-time.

Can a Majorana spinor, whose entries are real functions of the
space-time, describe the energy, linear and angular momentums of a
free spin one-half particle? We show that it can.

We show that the Majorana spinor is an irreducible
representation of the double cover of the proper
orthochronous Lorentz group and of the full Lorentz group.
The Fourier-Majorana and Hankel-Majorana
transforms are defined and related to the linear and angular momentums
of a free spin one-half particle.
\end{abstract}

\tableofcontents
\pagebreak
\section{Introduction}

In 1928 Paul Dirac published \textit{``The Quantum Theory of the
  Electron''} \cite{Dirac}, in which he introduced a relativistic
equation for the electron in interaction with an electromagnetic
potential, consisting of a complex 4x4 matrix
differential equation whose solution is a complex 4x1 column matrix,
called Dirac spinor, whose entries are complex functions of the
space-time. Using the algebra of the 4x4 matrices, he related the
electron's spin with the Lorentz group. He also noticed the existence
of negative-energy solutions which he used later in the prediction of
the existence of the anti-electron, the positron.

In 1937 Ettore Majorana published \textit{``A symmetric
theory of electrons and positrons''} \cite{majorana}, in which he
noted that \textit{``it is perfectly, and most
naturally, possible to formulate a theory of elementary neutral
particles which do not have negative (energy) states''}.
His work was based on the fact that there is a basis where the Dirac
equation for the free electron is a real, instead of complex, 4x4
matrix differential equation whose solution can be a real 4x1 column
matrix, called Majorana spinor, whose entries are real functions of
the space-time.
The existence of both positive and negative energy
solutions is a consequence of the extension, through the use of
complex numbers, of the free Dirac equation to include the
electromagnetic interaction. For neutral particles, the free Dirac
equation do not have to be extended in the same way it is when including
the electromagnetic interaction and, therefore, it
is possible to have a theory without negative energy solutions. 
Ettore Majorana disappeared in 1938.

There are applications of the Majorana's discovery in
theories trying to explain phenomena in neutrino physics, dark
matter searches, the fractional quantum Hall effect and superconductivity
\cite{solidstate}. 
There are good references on spinors \cite{pal, todorov, dreiner} and on its
relation with the Lorentz group \cite{pin}.  It is known (section 5 of
\cite{irreducible}) that the Majorana spinor is an irreducible
representation of the double cover of the proper orthochronous Lorentz
group. However, we could not find a study (without second quantization
operators) of the Majorana spinor solutions of the free Dirac
equation.

In the context of Clifford Algebras, the generalization of the Dirac
matrices algebra to other dimensions and metrics, there is
work on the geometric square roots of -1 
\cite{squareroot, hestenes_recent} and on the generalizations of the
Fourier transform \cite{clifford}, with applications to image
processing.

Our goal is to show that (without second quantization operators)
all the kinematic properties of a free spin 1/2
particle are present in the real solutions of the
real free Dirac equation.
In chapter 2 we define the Majorana matrices and spinors.
In chapter 3 we show that the Majorana spinor is an
irreducible representation of the double cover of the proper
orthochronous Lorentz group and of the full Lorentz group.
In chapter 4 we show the invariance of the free Dirac equation under
the action of the Lorentz group.
In 5 and 6 we define the Fourier-Majorana and Hankel-Majorana
transforms of a Majorana spinor whose entries are Lebesgue square
integrable real functions of the space coordinates.
In 7, by comparison with the particle/anti-particle solutions of the free
Dirac equation, we show that the Majorana transforms are related with the
linear and angular momentums of a free spin 1/2 particle.
In 8, we extend the Majorana transforms to include the energy.

\section{Majorana Matrices and Spinors}
The Majorana matrices, $i\gamma^\mu$ with $\mu=0,1,2,3$,
are the Dirac Gamma matrices, $\gamma^\mu$, times the imaginary unit.
The notation maintains explicit the relation between the Majorana and
Dirac Gamma matrices. 

\begin{defn}
$\mathbf{M}(m,n,\mathbb{F})$ is the set of $m\times n$ matrices whose
entries are elements of the field $\mathbb{F}$.
\end{defn}

\begin{defn}
The Majorana matrices, $i\gamma^\mu\in \mathbf{M}(4,4,\mathbb{C})$, are $4\times 4$ complex matrices
with anti-commutator $\{i\gamma^\mu,i\gamma^\nu\}$:
\begin{align}
(i\gamma^\mu)(i\gamma^\nu)+(i\gamma^\nu)(i\gamma^\mu)=-2g^{\mu\nu},\ \mu,\nu=0,1,2,3
\end{align}
Where $g=diag(1,-1,-1,-1)$ is the Minkowski metric.
The pseudo-scalar is $i\gamma^5\equiv
-\gamma^0\gamma^1\gamma^2\gamma^3$.
\end{defn}

The product of 2 Dirac Gamma
matrices is minus the
product of 2 corresponding Majorana matrices:
$\gamma^\mu\gamma^\nu=-i\gamma^\mu i\gamma^\nu$.

\begin{defn}
$\Gamma_-\equiv\{i\gamma^0,i\gamma^5,\gamma^0\gamma^5,i\gamma^5\gamma^0\gamma^j:
\ j=1,2,3\}$

$\Gamma_+\equiv\{1, \gamma^0\gamma^j,i\gamma^j,\gamma^5\gamma^j:\
j=1,2,3\}$

$\Gamma\equiv \Gamma_-\cup\Gamma_+$
\end{defn}

From the anti-commutator of the Majorana matrices, the
matrices in $\Gamma_\pm$ square respectively to $\pm 1$, and all
matrices in $\Gamma$ either commute or anti-commute with each other. 

\begin{defn}
The sets of matrices that (anti-)commute with a matrix $A\in \Gamma$ are:
$\Omega_\pm(A)=\{B\in\Gamma: AB=\pm BA\}$.
\end{defn}

\begin{prop}
\label{prop:anticommute}
The sets $\Omega_\pm(A)\cap \Gamma_+$ and $\Omega_\pm(A) \cap \Gamma_-$ are not
empty for all $A\in\Gamma\setminus\{1\}$.
\end{prop}

\begin{cor}
The matrices in $\Gamma\setminus\{1\}$ have null trace and
determinant 1.
\end{cor}

\begin{proof}
If $A\in\Gamma\setminus\{1\}$. Since there is $B\in\Omega_-(A)\cap
\Gamma_+$, we have $tr(A)=tr(BAB)=-tr(A)$.

Let $A\in \Gamma_-$. Since $A^2=-1$, then
$A=e^{\frac{\pi}{2}A}$ and $det(A)=e^{\frac{\pi}{2}tr(A)}=1$.

Let $A\in \Gamma_S\setminus\{1\}$.  Since there is $B\in\Omega_+(A)\cap
\Gamma_-$, we have $(AB)\in \Gamma_-$ and so $det(AB)=1$.
Since $det(B)=1$, then $det(A)=1$.
\end{proof}

\begin{prop}
\label{prop:basis}
$\Gamma$ is a basis
for the space of $4\times 4$ complex matrices.
\end{prop}
\begin{proof}
There are only 16 linearly independent $4\times 4$ complex matrices.

Let $B\equiv \sum_{i=1}^{16} a_iA_i$, where $a_i\in \mathbb{C}$ and
$A_i\in\Gamma$ are different elements of the set for each $i$.
We have $tr(A^\dagger_j B)=4a_j$, for $j=1,...,16$.
Then, $B=0$ implies that all the scalars $a_i$ are null 
and so all the elements in $\Gamma$ are linearly independent.
\end{proof}

\begin{prop}
\label{prop:commute}
For all commuting matrices
$A,B\in\Gamma\setminus\{1\}$, $AB=BA$: all matrices in
$\Gamma\setminus\{1,A,B,AB\}$ anti-commute with $A$ or
$B$. That is, $\Omega_-(A)\cup
\Omega_-(B)=\Gamma\setminus\{1,A,B,AB\}$.
\end{prop}

\begin{defn}
$\Gamma_2$ is the group of 32 Majorana matrices
products:
\begin{align}
\Gamma_2\equiv\{\pm 1,\ \pm i\gamma^\mu,\ \pm \gamma^0\gamma^j,
\ \pm i\gamma^5\gamma^0\gamma^j,\ \pm \gamma^\mu\gamma^5,\ \pm
i\gamma^5: \mu=0,1,2,3,\ j=1,2,3\}
\end{align}

\end{defn}

\begin{defn}
A $4\times 4$ representation of the
Majorana matrices, $M$, is a map from the Majorana matrices to the
space of $4\times 4$ complex matrices, verifying:
\begin{align}
\{M(i\gamma^\mu),M(i\gamma^\nu)\}&=-2g^{\mu\nu},\  
\mu,\nu=0,1,2,3
\end{align}
\end{defn}

\begin{prop}
\label{prop:similar}
Two $4\times 4$ representations of the Majorana matrices are
related by a similarity transformation, unique up to a complex factor \cite{diracmatrices}.
\end{prop}
\begin{proof}
Given a $4\times 4$ representation of the Majorana matrices,
$M$, we extend the domain from the Majorana matrices to $\Gamma$,
recursively, in such a way that for $k_1,k_2\in \Gamma_2$, if we know
$M(k_1)$ and $M(k_2)$, then $M(k_1k_2)\equiv M(k_1)M(k_2)$.

Let $A$ and $B$ be $4\times 4$ representations of the Majorana matrices.

We define the matrix $S$ as:
\begin{align}
S\equiv\sum_{g\in\Gamma_2}B(g^{-1}) S' A(g)
\end{align}
Where $S'$ will be defined later.

For all $h\in\Gamma$, it verifies $S A(h)=B(h)S$:
\begin{align}
&SA(h)=\sum_{g\in\Gamma_2}B(g^{-1})S' A(gh)\\
&=\sum_{l\in\Gamma_2}B(h l^{-1})S' A(l)=B(h)S
\end{align}
We define the matrix $T$ as:
\begin{align}
T\equiv \sum_{g\in\Gamma_2}A(g^{-1})T' B(g)
\end{align}

Where $T'$ will be defined later.
For all $h\in\Gamma$, it verifies $T B(h)=A(h)T$. Consequently, $T S
A(h)=A(h)T S$. 

Since $\gamma^\mu$ and $A(\gamma^\mu)$ obey to the same commutation
relations, the set $\{A(k), k\in\Gamma\}$ is also a basis for the
space of $4\times 4$ matrices. Therefore, $T S$ is equal to the identity matrix times a
coefficient. To check what the coefficient is:
\begin{align}
TS=\sum_{g\in\Gamma_2}A(g^{-1})T'S' A(g)
\end{align}

We choose $T'$, $S'$ such that $TS$ is non-null. Suppose
that such $T'$,$S'$ do not exist, then for all indexes $i,j$
$\sum_{g\in\Gamma_2}A_{ij}(g^{-1})A(g)=0$ which implies that the
matrices $A(g)$ are linear dependent, in contradiction with
Proposition \ref{prop:basis}. With a proper normalization, we can make $T=S^{-1}$.

Suppose that for all $h\in\Gamma$, $S'$ is invertible and also
verifies $S' A(h)=B(h)S'$. Then  $S'^{-1} SA(h)=A(h)S'^{-1} S$ and
again $S'^{-1} S$ must be proportional to the identity. 
Let $c\in \mathbb{C}$ be such that $S'^{-1}
S=c$. 
Multiplying on the left by $S'$, we get $S=c S'$.
\end{proof}

The Majorana matrices are themselves a $4\times 4$ representation of
the Majorana matrices. Therefore, choosing a $4\times 4$
representation of the Majorana matrices is the same as choosing a basis.

\begin{prop}
\label{prop:realsimilar}
Two $4\times 4$ unitary representations of the Majorana matrices are
related by an unitary similarity transformation, unique up to a phase.
\end{prop}
\begin{proof}
Let $A$ and $B$ be unitary representations of the Majorana matrices.
Then there is an invertible matrix $S$, unique up to a complex scalar,
such that $A(\gamma^\mu)S=S B(\gamma^\mu)$. Multiplying on the left by
$A^\dagger$ and on the right by $B^\dagger$ and making the hermitian
conjugate of the equation, we get $B(\gamma^\mu)S^\dagger=S^\dagger
A(\gamma^\mu)$.

So, for some complex $c$, $S^\dagger=c S^{-1}$. Applying the
determinant, we get $c=|det(S)|^2$ is real and positive. 
So, $(c^{-\frac{1}{2}}S)^\dagger(c^{-\frac{1}{2}}S)=1$.

Let both $S$ and $S'\equiv c S$, for some complex $c$, be
unitary. Then, $(cS)^\dagger(c S)=|c|^2=1$,
so $c=e^{i\theta}$ for some real $\theta$.
\end{proof}

In the Majorana bases, the Majorana matrices are $4\times 4$ real
orthogonal matrices. An example of the Majorana matrices in a particular Majorana basis is:
\begin{align}
\begin{array}{llllll}
\label{basis}
i\gamma^1=&\left[ \begin{smallmatrix}
+1 & 0 & 0 & 0 \\
0 & -1 & 0 & 0 \\
0 & 0 & -1 & 0 \\
0 & 0 & 0 & +1 \end{smallmatrix} \right]&
i\gamma^2=&\left[ \begin{smallmatrix}
0 & 0 & +1 & 0 \\
0 & 0 & 0 & +1 \\
+1 & 0 & 0 & 0 \\
0 & +1 & 0 & 0 \end{smallmatrix} \right]&
i\gamma^3=\left[ \begin{smallmatrix}
0 & +1 & 0 & 0 \\
+1 & 0 & 0 & 0 \\
0 & 0 & 0 & -1 \\
0 & 0 & -1 & 0 \end{smallmatrix} \right]\\
\\
i\gamma^0=&\left[ \begin{smallmatrix}
0 & 0 & +1 & 0 \\
0 & 0 & 0 & +1 \\
-1 & 0 & 0 & 0 \\
0 & -1 & 0 & 0 \end{smallmatrix} \right]&
i\gamma^5=&\left[ \begin{smallmatrix}
0 & -1 & 0 & 0 \\
+1 & 0 & 0 & 0 \\
0 & 0 & 0 & +1 \\
0 & 0 & -1 & 0 \end{smallmatrix} \right]&
=-\gamma^0\gamma^1\gamma^2\gamma^3
\end{array}
\end{align}
\begin{prop}
\label{prop:realsimilar}
Two $4\times 4$ real representations of the Majorana matrices are
related by a real similarity transformation, unique up to a real factor.
\end{prop}
\begin{proof}
Let $A$ and $B$ be real representations of the Majorana matrices.
Then there is an invertible matrix $S$, unique up to a complex factor,
such that $A(\gamma^\mu)S=S B(\gamma^\mu)$. Conjugating the equation,
we get that, for some complex $c$, $S^*=c S$. Applying the
module of the determinant, we get $c=e^{i\theta}$ for some real $\theta$ and $(e^{i\frac{\theta}{2}}S)^*=(e^{i\frac{\theta}{2}}S)$. 
\end{proof}

\begin{defn}
The Dirac spinor is a $4\times 1$ complex column matrix, $\mathbf{M}(4,1,\mathbb{C})$.
\end{defn}

The space of Dirac spinors is a 4 dimensional complex vector space.

\begin{defn}
Let $S$ be an invertible matrix such that $S i\gamma^\mu S^{-1}$ is
real, for $\mu=0,1,2,3$.

The set of Majorana spinors, $Pinor$, is the set of Dirac spinors
verifying the Majorana condition:
\begin{align}
Pinor\equiv \{u\in \mathbf{M}(4,1,\mathbb{C}): S^{-1} S^* u^*=u\}
\end{align}
Where $^*$ denotes complex conjugation.
\end{defn}

\begin{rem}
Let $W$ be a subset of a vector space $V$ over $\mathbb{C}$. 
$W$ is a real vector space iff:

1) $0\in W$;

2) If $u,v\in W$, then $u+v\in W$;

3) If $u\in W$ and $c\in \mathbb{R}$, then $c u\in W$.
\end{rem}

From the previous remark, the set of Majorana spinors is a 4
dimensional real vector space. Note that the linear combinations of
Majorana spinors with complex scalars do not verify the Majorana
condition. The Majorana spinor, in the Majorana bases, is a $4\times 1$
real column matrix.

\begin{defn}
$End(Pinor)$ is the set of endomorphisms of Majorana spinors, that is,
the set of linear maps from and to Majorana spinors.
\end{defn}

$End(Pinor)$ is a 16
dimensional real vector space, generated by the linear
combinations with real scalars of the 16 matrices in the basis
$\Gamma$. In the Majorana bases, $End(Pinor)=\mathbf{M}(4,1,\mathbb{R})$.

\section{Majorana representation of the Lorentz group}
\subsection{Double cover of the Lorentz group}
We define some symbols for the sets we will use: 
\begin{defn}
$GL(n,\mathbb{F})$ is the group of $n\times
n$ invertible matrices over the field $\mathbb{F}$. 

$SL(n,\mathbb{F})$ is the group of $n\times n$ invertible matrices
over the field $\mathbb{F}$ with determinant $1$.

$O(n)$ is the group of $n\times
n$ real orthogonal matrices.

$SO(n)$ is the group of $n\times
n$ real orthogonal matrices with determinant 1.

$SPD(n)$ is the set  of $n\times
n$ real symmetric positive definite matrices.
\end{defn}

\begin{defn} The set of Lorentz matrices, $O(1,3)\equiv\{\Lambda \in
  \mathbf{M}(4,4,\mathbb{R}): \Lambda^T g \Lambda=g \}$, is the set of
  real matrices that leave the metric, $g=diag(1,-1,-1,-1)$, invariant.
\end{defn}

\begin{defn}
In a basis where the Majorana matrices are unitary, the set $Maj$ is
defined as:
\begin{align}
Maj\equiv\{M\in End(Pinor): (i\gamma^5) M (-i\gamma^5)=-M,\ (i\gamma^0) M (-i\gamma^0) &=-M^\dagger\}
\end{align}
\end{defn}

The only matrices in $\Gamma$ that are also in $Maj$ are the Majorana
matrices, $i\gamma^\mu$, therefore $Maj$ is the 4 dimensional real space of the linear
combinations with real coefficients of Majorana matrices.

\begin{defn}
$Pin(3,1)$ \cite{pin} is the set of endomorphisms of Majorana
spinors which leave the space $Maj$ invariant, that is:
\begin{align}
Pin(3,1)\equiv \Big\{S\in End(Pinor):\ |det(S)|=1,\ S^{-1}(i\gamma^\mu)S\in Maj,\
\mu=0,1,2,3 \Big\}
\end{align}
\end{defn}

\begin{prop}
\label{prop:map}
The map $\Lambda:Pin(3,1)\to O(1,3)$ defined by:
\begin{align}
(\Lambda(S))^\mu_{\ \nu}i\gamma^\nu\equiv S^{-1}(i\gamma^\mu)S
\end{align}
is two-to-one and surjective.
\end{prop}

\begin{proof}
1) Let $S\in Pin(3,1)$. Since the Majorana matrices are a basis of the
real vector space $Maj$, there is an unique real matrix $\Lambda(S)$ such that:
\begin{align}
(\Lambda(S))^\mu_{\ \nu}i\gamma^\nu=S^{-1}(i\gamma^\mu)S
\end{align}
Therefore, $\Lambda$ is a map with domain $Pin(3,1)$. Now we can check
that $\Lambda(S)\in O(1,3)$:
\begin{align}
&(\Lambda(S))^\mu_{\ \alpha}g^{\alpha\beta}(\Lambda(S))^\nu_{\
  \beta}=-\frac{1}{2}(\Lambda(S))^\mu_{\
  \alpha}\{i\gamma^\alpha,i\gamma^\beta\}(\Lambda(S))^\nu_{\
  \beta}=\\
&=-\frac{1}{2}S\{i\gamma^\mu,i\gamma^\nu\}S^{-1}=Sg^{\mu\nu}S^{-1}=g^{\mu\nu}
\end{align}
We have proved that $\Lambda$ is a map from $Pin(3,1)$ to $O(1,3)$.

2) Since any $\lambda\in O(1,3)$ conserve the metric, the matrices
$M(i\gamma^\mu)\equiv \lambda^\mu_{\ \nu} i\gamma^\nu$ are a representation of the Majorana matrices:
\begin{align}
\{M(i\gamma^\mu),M(i\gamma^\nu)\}=-2\lambda^\mu_{\ \alpha}g^{\alpha\beta}\lambda^\nu_{\ \beta}=-2(\lambda g\lambda^T)^{\mu\nu}=-2g^{\mu\nu}
\end{align}
In a basis where the Majorana matrices are real, from
\ref{prop:realsimilar} there is a real invertible matrix $S_\Lambda$, 
unique up to a real factor, such that $\lambda^\mu_{\ \nu} i\gamma^\nu=S^{-1}_\lambda
(i\gamma^\mu)S_\lambda$.
Setting $|det(S)|=1$ we fix the real factor up to a signal $\pm
1$.
Therefore, $\pm S_\lambda\in Pin(3,1)$ and we proved that the map
$\Lambda:Pin(3,1)\to O(1,3)$ is two-to-one and surjective.
\end{proof}

\begin{lem}
$Pin(3,1)=Pin'(3,1)$, where $Pin'(3,1)$ is, in a basis where the
Majorana matrices are unitary:
\begin{align}
Pin'(3,1)\equiv \Big\{S\in End(Pinor):& (i\gamma^5)S=a
S(i\gamma^5),\\
&(i\gamma^0)S=b S^{-1\dagger}(i\gamma^0),\\
& |det(S)|=1;\ a,b\in \{-1,1\}\Big\}
\end{align}
\end{lem}

\begin{proof}
1) For all $S\in Pin'(3,1)$, $S^{-1}(i\gamma^\mu)S\in Maj$ and so
$Pin'(3,1)\subset Pin(3,1)$.
 
2) In a basis where the Majorana matrices are unitary, 
for all $S\in Pin(3,1)$, since $S^{-1}(i\gamma^\mu)S\in Maj$, we have:
\begin{align}
(i\gamma^5) S^{-1}(i\gamma^\mu)S (-i\gamma^5)
&=-S(i\gamma^\mu)S=S^{-1}(i\gamma^5)(i\gamma^\mu)(-i\gamma^5)S\\
(i\gamma^0) S^{-1}(i\gamma^\mu)S (-i\gamma^0)
&=-S^\dagger(i\gamma^\mu)^\dagger S^{-1\dagger}=S^\dagger(i\gamma^0)(i\gamma^\mu)(-i\gamma^0) S^{-1\dagger}
\end{align}
On the other hand:
\begin{align}
(i\gamma^5) S^{-1}(i\gamma^\mu)S (-i\gamma^5)
&=(-\Lambda(S))^\mu_{\ \nu}(i\gamma^\nu)\\
(i\gamma^0) S^{-1}(i\gamma^\mu)S (-i\gamma^0)
&=(\Lambda(S)g)^\mu_{\ \nu} (i\gamma^\nu)
\end{align}
We can easily check that $(-\Lambda),(\Lambda g)\in O(1,3)$. 
From proposition \ref{prop:map}, we get that the matrices in
$Pin(3,1)$ corresponding to $(-\Lambda),(\Lambda g)\in O(1,3)$ are
unique up to a sign. Therefore, $Pin(3,1)\subset Pin'(3,1)$.
\end{proof}

\begin{defn}
In a basis where the Majorana matrices are unitary, the subset
$Spin^+(3,1)\subset Pin(3,1)$ is:
\begin{align}
Spin^+(3,1)\equiv \Big\{S\in End(Pinor):& (i\gamma^5)S=S(i\gamma^5),\\
&(i\gamma^0)S=S^{-1\dagger}(i\gamma^0),\\
&|det(S)|=1 \Big\}
\end{align}
\end{defn}

\begin{prop}
1) $Pin(3,1)$ and $Spin^+(3,1)$ are groups.

2) In a basis where the Majorana matrices are unitary, if 
$S\in Pin(3,1)$ ($Spin^+(3,1)$) then $S^\dagger\in Pin(3,1)$ ($Spin^+(3,1)$).
\end{prop}
\begin{proof}
$Pin(3,1)$ and $Spin^+(3,1)$ are subsets of the group $SL(4,\mathbb{C})$.
They include the identity matrix, $1\in Pin(3,1),Spin^+(3,1)$.

Let $S_\pm\in Pin(3,1)$. Then, in a basis where the Majorana
matrices are unitary, for some $a_\pm,b_\pm\in \{-1,1\}$:
\begin{align}
(i\gamma^5)S_\pm=a_\pm S_\pm(i\gamma^5),\ (i\gamma^0)S_\pm=b_\pm S_\pm^{-1\dagger}(i\gamma^0)
\end{align}
Making the inverse (hermitian conjugate) of the equation on the left and the
hermitian conjugate (inverse) of the equation on the right we get:
\begin{align}
-S^{-1}_\pm(i\gamma^5)=-a_\pm(i\gamma^5)S_\pm^{-1}&,\
-S^{\dagger}_\pm(i\gamma^0)=-b_\pm(i\gamma^0)S_\pm^{-1}\\
-S^{\dagger}_\pm(i\gamma^5)=-a_\pm(i\gamma^5)S_\pm^{\dagger}&,\
-S^{-1}_\pm(i\gamma^0)=-b_\pm(i\gamma^0)S_\pm^{\dagger}
\end{align}
Therefore, $S_\pm^\dagger\in Pin(3,1)$ and the product $S_+ S_-^{-1}\in  Pin(3,1)$:
\begin{align}
(i\gamma^5)S_+ S_-^{-1}&=(a_+a_-)S_+ S_-^{-1}(i\gamma^5)\\
(i\gamma^0)S_+ S_-^{-1}&=(b_+b_-)S_+^{-1\dagger} S_-^{\dagger}(i\gamma^0)
\end{align}

In the particular case $S_\pm\in Spin^+(3,1)$, we have
$a_\pm,b_\pm=1$. Then $S_\pm^\dagger\in Spin^+(3,1)$ and the product $S_+ S_-^{-1}\in  Spin^+(3,1)$. 
\end{proof}

\begin{defn}
The discrete pin subgroup $\Delta\subset Pin(3,1)$ is:
\begin{align}
\Delta\equiv \{\pm 1,\pm i\gamma^0,\pm \gamma^0\gamma^5,\pm
i\gamma^5\}
\end{align}
\end{defn}

\begin{lem}
For all $S\in Pin(3,1)$, there are only two factors $\pm d\in
\Delta$ and 
correspondingly only two $\pm S'\in Spin^+(3,1)$, such that
$S=(\pm d)(\pm S')$.
\end{lem}

\begin{proof}
Let $S\in Pin(3,1)$ and  $a,b\in\{-1,1\}$ be such that, in a basis where the Majorana
matrices are unitary:
\begin{align}
(i\gamma^5)S=a S(i\gamma^5),\ (i\gamma^0)S=b S^{-1\dagger}(i\gamma^0)
\end{align}

There are always only two factors $\pm d\in \Delta$, such that $d^{-1}S \in Spin^+(3,1)$:
\begin{align}
a=b=1,\ &d=\pm 1\\
a=-b=1,\ &d=\pm (i\gamma^5)\\
-a=b=1,\ &d=\pm (i\gamma^0)\\
-a=-b=1,\ &d=\pm (\gamma^0\gamma^5)
\end{align}
\end{proof}

\begin{rem}
\label{rem:factor}
1) Every real invertible matrix can be uniquely factored as the product
of an orthogonal matrix and a symmetric positive definite matrix.

2) For all real symmetric positive definite matrix $\Pi$, there is an unique
symmetric matrix $B$ such that $\Pi=e^B$.

3) For all real orthogonal matrix with determinant $1$, $\Theta$, there is a
skew-symmetric matrix $A$ such that $\Theta=e^A$.
\end{rem}

\begin{lem}
\label{lem:SpinProperties}
$Spin^+(3,1)=Spin'^+(3,1)$, where:
\begin{align}
Spin'^+(3,1)\equiv\{e^{\theta^j
  i\gamma^5\gamma^0\gamma^j}e^{b^j\gamma^0\gamma^j}: \theta^j,b^j\in
\mathbb{R},\ j=1,2,3\}
\end{align}
Note that there is a sum in the index $j$.
\end{lem}

\begin{proof}
In a Majorana basis, since $S$ is invertible and real, 
from point 1) in remark \ref{rem:factor},
there is an unique $\Theta\in O(4)$ and unique $\Pi\in SPD(4)$
such that $S=\Theta\Pi$.

From point 2) in remark \ref{rem:factor}, there is 
an unique symmetric $B$ such that $\Pi=e^B$.

Since $S,S^\dagger\in Spin^+(3,1)$, also $S^\dagger S=e^{2B}\in
Spin^+(3,1)$ and we have:
\begin{align}
(i\gamma^5)e^{2B}=e^{2B}(i\gamma^5),\ (i\gamma^0)e^{2B}=e^{-2B}(i\gamma^0)
\end{align}
From the uniqueness of $B$ we get $(i\gamma^5)B=B(i\gamma^5)$ and
$(i\gamma^0)B=-B(i\gamma^0)$. 
In a Majorana basis, the only symmetric matrices in $\Gamma$
satisfying the previous equations are $\gamma^0\gamma^j$,
$j=1,2,3$. Therefore, there are unique $b^j\in \mathbb{R}$, $j=1,2,3$,
such that $\Pi=e^{b^j \gamma^0\gamma^j}$. Since $\gamma^0\gamma^j$ is
traceless, $det(\Pi)=1$ and $\Pi\in Spin^+(3,1)$.

Since $det(S)=det(\Pi)=1$, also $det(\Theta)=1$. We can write:
\begin{align}
(i\gamma^5)\Theta e^B=\Theta e^B(i\gamma^5),\ (i\gamma^0)\Theta e^B=\Theta e^{-B}(i\gamma^0)
\end{align}
Multiplying the equations by $e^{-B}$ on the right, also $\Theta\in Spin^+(3,1)$:
\begin{align}
(i\gamma^5)\Theta=\Theta (i\gamma^5),\ (i\gamma^0)\Theta=\Theta (i\gamma^0)
\end{align}
From point 3) in remark \ref{rem:factor}, there is 
a skew-symmetric $A$ such that $\Theta=e^A$. In a Majorana basis, the
only skew-symmetric matrices in $\Gamma$ are in the commuting sets
$\{i\gamma^0,\gamma^0\gamma^5,i\gamma^5\}$ and $\{i\gamma^5\gamma^0\gamma^j:\ j=1,2,3\}$. 
Therefore, there are $\theta^j,a^j\in \mathbb{R}$, $j=1,2,3$,
such that $C\equiv a^1 i\gamma^0+a^2 \gamma^0\gamma^5+a^3 i\gamma^5$
and $\Theta=e^{C}e^{\theta^j i\gamma^5\gamma^0\gamma^j}$. We can write:
\begin{align}
(i\gamma^5)e^{C}e^{\theta^j
  i\gamma^5\gamma^0\gamma^j}=e^{C}e^{\theta^j i\gamma^5\gamma^0\gamma^j}(i\gamma^5),
\ (i\gamma^0)e^{C}e^{\theta^j
  i\gamma^5\gamma^0\gamma^j}=e^{C}e^{\theta^j i\gamma^5\gamma^0\gamma^j}(i\gamma^0)
\end{align}
Multiplying the equations by $e^{-\theta^j i\gamma^5\gamma^0\gamma^j}$
on the right we get:
\begin{align}
(i\gamma^5)e^C=e^C (i\gamma^5),\ (i\gamma^0)e^C=e^C (i\gamma^0)
\end{align}
The last equations imply that $e^{a^1 i\gamma^0+a^2
  \gamma^0\gamma^5+a^3 i\gamma^5}=e^{-a^1 i\gamma^0+a^2
  \gamma^0\gamma^5-a^3 i\gamma^5}$ and so 
$(i\gamma^j)e^C=e^C (i\gamma^j)$, for $j=1,2,3$. Then, $e^C$ commutes with all
matrices in $\Gamma$ and so it must be proportional to the
identity. From $det(e^C)=1$ we get that $e^C=\pm 1$. 

If $e^C=-1$, the signal can be absorbed. We define
$|\theta|\equiv \sqrt{\theta^j\theta^j}$. If $|\theta|$ is null, then $\Theta=e^{\pi
  i\gamma^5\gamma^0\gamma^1}$; if not, then 
$\Theta=e^{(1+\frac{\pi}{|\theta|})\theta^j i\gamma^5\gamma^0\gamma^j}$.

Checking that $e^{\theta^j
  i\gamma^5\gamma^0\gamma^j}, e^{b^j\gamma^0\gamma^j}\in Spin^+(3,1)$
for all $\theta^j, b^j\in \mathbb{R},\
j=1,2,3$, we have completed the prove.
\end{proof}

\begin{prop} 
$O(1,3)$ is a group, the Lorentz group, and the map
$\Lambda:Pin(3,1)\to O(1,3)$, defined in proposition \ref{prop:map},
is a group homomorphism.
\end{prop}
\begin{proof}
The matrix product is associative and $\Lambda(1)=1\in O(1,3)$.

For all $S,S'\in Pin(3,1)$, we have:
\begin{align}
(\Lambda(S S'))^\mu_{\ \beta}i\gamma^\beta&=S'^{-1}
S(i\gamma^\mu)S S'=(\Lambda(S))^\mu_{\ \alpha}S'^{-1}
i\gamma^\alpha S'\\
&=(\Lambda(S))^\mu_{\ \alpha}(\Lambda(S'))^\alpha_{\ \beta}i\gamma^\beta
\end{align}
This implies that $\Lambda^{-1}(S)=\Lambda(S^{-1})\in O(1,3)$ and
$\Lambda(S)\Lambda(S')=\Lambda(S S')\in O(1,3)$. Since the map $\Lambda$ is
surjective, then $O(1,3)$ is a group and the map $\Lambda$ is a group homomorphism.
\end{proof}

\begin{defn}
The proper orthochronous Lorentz group $SO^+(1,3)$ is:
\begin{align}
SO^+(1,3)\equiv\{\Lambda(S): S\in Spin^+(1,3)\}
\end{align}
Where $\Lambda$ is the map defined in proposition \ref{prop:map}.
\end{defn}

Since there is a two-to-one surjective group homomorphism,
$Pin(3,1)$ is a double cover of $O(1,3)$ and $Spin^+(3,1)$ 
is a double cover of $SO^+(1,3)$.

In a Majorana basis, by identifying $i$ with $i\gamma^5$ and $\gamma^0\gamma^j$
with the Pauli matrices $\sigma^j$, we can see that $Spin^+(3,1)$ is
isomorphic to $SL(2,\mathbb{C})$.

\subsection{Majorana Spinor representation}
\begin{defn}
A representation $(M_G,V)$ of a group $G$ is defined by:

1) the representation space $V$, which is a vector space;

2) the representation map $M:G\to GL(V)$ from the
group elements to the automorphisms
of the representation space, verifying for $\Lambda_1,\Lambda_2\in G$:
\begin{align}
M(\Lambda_1)M(\Lambda_2)=M(\Lambda_1\Lambda_2)
\end{align}
\end{defn}

\begin{defn}
The Majorana spinor representation of $Pin(3,1)$ is defined by:

1) the representation space $V=Pinor$ is the space of Majorana spinors;

2) In a basis where the Majorana matrices are unitary, the representation map is:
\begin{align}
M(S)=S,\ S\in Pin(3,1)
\end{align}
\end{defn}

The Majorana spinor representation of the subgroup $Spin^+(3,1)\subset
Pin(3,1)$ is obtained from
the representation of $Pin(3,1)$ by restricting the domain of the
representation map to the subgroup $Spin^+(3,1)\subset Pin(3,1)$. 

\begin{defn}
Let $W$ be a subspace of $V$. $(M_G,W)$ is a subrepresentation of $(M_G,V)$
if $W$ is invariant under the group action, that is, for all $w\in W$:
$(M(g) w)\in W$, for all $g\in G$.
\end{defn}

\begin{defn}
$W^\bot$ is the orthogonal complement of the subspace $W$ of the
vector space $V$ if:

1) all $v\in V$ can be expressed as $v=w+x$, where $w\in W$ and $x\in
W^\bot$; 

2) if $w\in W$ and $x\in W^\bot$, then $x^\dagger w=0$.
\end{defn}

\begin{defn}
The representation $(M_G,V)$ is semi-simple if for all subrepresentation
$(M_G,W)$ of $(M_G,V)$ , $(M_G,W^\bot)$ is also a subrepresentation of
$(M_G,V)$, where $W^\bot$ is the orthogonal complement of the subspace $W$.
\end{defn}

\begin{lem}
\label{lem:orthogonal}
Consider a representation $(M_G,V)$ of a group $G$.
For all $g\in G$, if there is $h\in G$ such that
$M(h)=M^\dagger(g)$, then the representation $(M_G,V)$ is semi-simple.
\end{lem}

\begin{proof} 
Let $(M_G,W)$ be a subrepresentation of $(M_G,V)$.
$W^\bot$ is the orthogonal complement of $W$.

For all $x\in W^\bot$, $w\in W$ and $g\in G$, 
$(M(g)x)^\dagger w=x^\dagger (M^\dagger(g)w)$. 

Since $W$ is invariant and there is $h\in G$, such
that $M(h)=M^\dagger(g)$, then $w'\equiv (M^\dagger(g)w)\in W$.
 
Since $x\in W^\bot$ and $w'\in W$, then $x^\dagger w'=0$.

This implies that if $x$ is in the orthogonal complement of $W$ ($x\in W^\bot$), also
$M(g)x$ is in the orthogonal complement of $W$ ($M(g)x\in W^\bot$), for all $g\in G$.
\end{proof}

\begin{prop}
\label{prop:MajoranaSemisimple}
The Majorana spinor representation of $Spin^+(3,1)$ is semi-simple.
\end{prop}

\begin{proof}
From point 1) in lemma \ref{lem:SpinProperties} and lemma \ref{lem:orthogonal}.
\end{proof}

\begin{defn}
A representation $(M_G,V)$ is irreducible if their only 
sub-representations are the trivial sub-representations: 
$(M_G,V)$ and $(M_G,\{0\})$, where $\{0\}$ is
the null space.
\end{defn}

\begin{lem}
\label{lem:commuting}
Consider a semi-simple representation $(M_G,V)$ of a group $G$.
If the set of hermitian automorphisms of $V$ that
square to $1$ and commute with $M(g)$, for all $g\in G$, is $\{+1,-1\}$, 
then the representation $(M_G,V)$ is irreducible ($1$ is the identity matrix).
\end{lem}

\begin{proof}
Let $(M_G,W)$ and $(M_G,W^\bot)$ be sub-representations of $(M_G,V)$,
where $W^\bot$, the orthogonal complement of $W$.

There is an automorphism $P: V\to V$, such that, 
for $w,w'\in W$, $x,x'\in W^\bot$, $P(w+x)=(w-x)$. $P^2=1$ and $P$ is
hermitian:
\begin{align}
&(w'+x')^\dagger (P (w+x))=w'^\dagger w-x'^\dagger x=
(P(w'+x'))^\dagger (w+x)
\end{align}
Let $w'\equiv M(g) w\in W$ and $x'\equiv M(g) x\in
W^\bot$:
\begin{align}
M(\Lambda)P(w+x)&=M(\Lambda)(w-x)=(w'-x')\\
P M(\Lambda)(w+x)&=P(w'+x')=(w'-x')
\end{align}
Which implies that $P$ commutes with $M(g)$ for all $g\in G$.

If $P=+1$, then $W=V$:
\begin{align}
+(w+x)=P(w+x)=(w-x) \implies x=0
\end{align}
If $P=-1$, then $W$ is the null space:
\begin{align}
-(w+x)=P(w+x)=(w-x) \implies w=0
\end{align}
\end{proof}

\begin{prop}
\label{prop:irreducible}
The Majorana spinor representation of $Spin^+(3,1)$ is irreducible.
\end{prop}

\begin{proof}
The hermitian linear transformations from
and to Majorana spinors, are generated by the linear combinations with
real coefficients of the $10$ matrices in the basis
$\Gamma_S\equiv\{1, \gamma^0\gamma^j,i\gamma^j,\gamma^5\gamma^j\}$, where $j=1,2,3$.

The only matrix in $\Gamma_S$ commuting
with all $S\in Spin^+(3,1)$ is the identity
matrix. Therefore, the set of hermitian automorphisms 
of the Majorana spinors that square to $1$ and commute with
all $S\in Spin^+(3,1)$, is $\{+1,-1\}$. Applying
proposition \ref{prop:MajoranaSemisimple} and lemma
\ref{lem:commuting} the proposition is proved.
\end{proof}

The Majorana spinor representation of the group $Pin^+(3,1)$ is
also irreducible because it is already irreducible for the subgroup
$Spin^+(3,1)\subset Pin^+(3,1)$.

\section{Majorana spinor solutions of the free Dirac equation}
\begin{defn}
$L^2(\mathbb{R}^n)$ is the Hilbert space of real functions of $n$ real
variables whose square is Lebesgue integrable in $\mathbb{R}^n$.
The internal product is:
\begin{align}
<f,g>\equiv\int d^nx f(x)g(x),\ f,g \in L^2(\mathbb{R}^n)
\end{align} 
\end{defn}

\begin{rem}
\label{rem:fourier}
If $f\in L^2(\mathbb{R}^n)$, then $f_s, f_c\in L^2(\mathbb{R}^n)$:
\begin{align}  
f_c(p)&\equiv\int d^nx\ \cos(p\cdot x)f(x)\\
f_s(p)&\equiv \int d^nx\ \sin(p\cdot x)f(x)
\end{align}
The Dirac delta $\delta^n$ is a well defined operator of the Hilbert
space $L^2(\mathbb{R}^n)$:
\begin{align}
\delta^n(x)&\equiv \int \frac{d^np}{(2\pi)^n} \cos(p\cdot x)\\ 
f(0)&=\int d^nx\ \delta^n(x)f(x)
\end{align}
The domain of integration is $\mathbb{R}^n$.
\end{rem}

\begin{rem}
\label{rem:derivative}
The derivative $\partial_i$, $i=1,...,n$, is a skew-symmetric
operator of the Hilbert space $L^2(\mathbb{R}^n)$:
\begin{align}
\int d^nx (\partial_i f(x))g(x)=-\int d^nx
f(x)(\partial_i g(x)),\ f,g \in L^2(\mathbb{R}^n)
\end{align} 
\end{rem}

The free Dirac equation is:
\begin{align}
(i\gamma^\mu \partial_\mu-m)\Psi(x)=0
\end{align}
We are looking for solutions where $\Psi(x)\in Pinor\otimes
L^2(\mathbb{R}^4)$ is a Majorana
spinor, whose entries are square integrable functions of the space-time.

If we change $x^\mu\to (\Lambda(S))^\mu_{\ \nu}
x^\nu$ and $\Psi(x)\to S\Psi(x)$, where $S\in Pin(3,1)$, 
and we multiply the equation on the left by $S^{-1}$ we get:

\begin{align}
&S^{-1}(i\gamma^\mu g_{\mu\nu}  (\Lambda(S))^\nu_{\
  \alpha}g^{\alpha\beta}\partial_\beta-m)S\Psi(x)=\\
&=(i\gamma^\delta (\Lambda(S))^\mu_{\ \delta} g_{\mu\nu}  (\Lambda(S))^\nu_{\
  \alpha}g^{\alpha\beta}\partial_\beta-m)\Psi(x)=\\
&=(i\gamma^\delta \partial_\delta-m)\Psi(x)=0
\end{align}
The equation stays invariant.

If we multiply by $-i\gamma^0$ on the left, the equation can be rewritten as:
\begin{align}
(\partial_0+iH(\vec{x}))\Psi(x)&=0\\
iH(\vec{x})\equiv \gamma^0\gamma^j \partial_j-i\gamma^0 m&,\ j=1,2,3
\end{align}
The solution is:
\begin{align}
\Psi(x)=e^{-iH(\vec{x}) x^0}\psi(\vec{x})
\end{align}
Where $\psi(\vec{x})\in Pinor\otimes
L^2(\mathbb{R}^3)$ is a Majorana
spinor, whose entries are square integrable functions of the space coordinates.

Now we can write $\psi(\vec{x})=M(\vec{x})\chi$, where  
$M(\vec{x})\in End(Pinor)\otimes L^2(\mathbb{R}^3)$ is a Majorana
spinor endomorphism, whose entries are square integrable functions of
the space and $\chi\in Pinor$ is a Majorana spinor.
Suppose that for some $E\in \mathbb{R}$, $M$ verifies the equation:
\begin{align}
iH(\vec{x})M(\vec{x})=M(\vec{x})i\gamma^0E
\end{align} 
In the next two sections we will see that these matrices satisfying
the above equation have interesting properties. Now we have:
\begin{align}
\Psi(x)=M(\vec{x})e^{-i\gamma^0 E x^0}\chi
\end{align}

Before moving to the next section, we will fix notation.
If $p, q$ are Lorentz vectors, we define $\slashed p=\gamma^\mu p_\mu$ and $p\cdot q=p^\mu q_\mu$.
Given a mass $m\geq 0$, we define:
\begin{align}
\vec{p}^j&=p^j,\ j=1,2,3\\
\vec{\slashed p}&=\vec{\gamma}\cdot\vec{p}\\
E_p&=\sqrt{\vec{p}^2+m^2}\\
\slashed p&=\gamma^0 p^0-\vec{\gamma}\cdot\vec{p}
\end{align}

\section{Linear Momentum of Majorana spinors}
\begin{defn}
$L^2_4(\mathbb{R}^n)$ is the Hilbert space $Pinor\otimes L^2(\mathbb{R}^n)$, that is, Majorana spinors whose
entries are square integrable functions of $\mathbb{R}^n$. The internal product is:
\begin{align}
<\Psi,\Phi>\equiv\int d^n x\ \Psi^\dagger(x)\Phi(x),\ \Psi,\Phi \in L^2_4(\mathbb{R}^n)
\end{align} 
\end{defn}

\begin{defn}
The Fourier-Majorana Transform $\psi(\vec{p})$ of a Majorana spinor
$\Psi(\vec{x})\in L^2_4(\mathbb{R}^3)$ is the Majorana spinor:
\begin{align}
\psi(\vec{p})&\equiv\int d^3\vec{x}\ O(\vec{p},\vec{x})\Psi(\vec{x})\\
O(\vec{p},\vec{x})&\equiv e^{-i\gamma^0\vec{p} \cdot
\vec{x}}\frac{\slashed p \gamma^0+m}{\sqrt{E_p+m}\sqrt{2E_p}}
\end{align}
Where $m\geq 0$ $p^0=E_p=\sqrt{\vec{p}^2+m^2}$. 
\end{defn}

\begin{prop}
The Fourier-Majorana Transform $\psi(\vec{p})$ of a Majorana spinor
$\Psi(\vec{x})\in L^2_4(\mathbb{R}^3)$ is also in the Hilbert space $L^2_4(\mathbb{R}^3)$.
\end{prop}

\begin{proof}
In the Majorana bases, $O(\vec{p},\vec{x})$ and $\Psi(\vec{x})$ are
real and so is $\psi(\vec{p})$.

We have:
\begin{align}
&|[\frac{\slashed p \gamma^0+m}{\sqrt{E_p+m}\sqrt{2E_p}}]_{ij}|^2\leq
\frac{E_p+m}{2E_p}\leq 1,\ i,j=1,2,3,4\\ 
&|\psi_i(\vec{p})|^2\leq \sum_{j=1}^4|\int d^3\vec{x}\ \cos(\vec{p}\cdot\vec{x})\Psi_j(\vec{x})|^2+|\int d^3\vec{x}\ \sin(\vec{p}\cdot\vec{x})\Psi_j(\vec{x})|^2
\end{align}
From remark \ref{rem:fourier}, we have that both $\int d^3\vec{x}\
\cos(\vec{p}\cdot\vec{x})\Psi_j(\vec{x})$
and $\int d^3\vec{x}\
\sin(\vec{p}\cdot\vec{x})\Psi_j(\vec{x})$ are square integrable and
therefore $|\psi_i(\vec{p})|^2$ is square integrable.
\end{proof}

\begin{prop}
The inverse Fourier-Majorana transform of $\psi(\vec{p})$ is:
\begin{align}
\Psi(\vec{x})&=\int \frac{d^3\vec{p}}{(2\pi)^3}
O^\dagger(\vec{p},\vec{x})\psi(\vec{p})\\
O^\dagger(\vec{p},\vec{x})&=\frac{\slashed p \gamma^0+m}{\sqrt{E_p+m}\sqrt{2E_p}}e^{i\gamma^0\vec{p} \cdot
\vec{x}}
\end{align}
$O^\dagger$ is the hermitian conjugate of $O$.
\end{prop}

\begin{proof}
The matrix $O^\dagger(\vec{p},\vec{x})$ verifies:
\begin{align}
O^\dagger(\vec{p},\vec{x})&=\frac{\slashed
  p}{m}O^\dagger(\vec{p},\vec{x})\gamma^0\\
i\gamma^0(i\vec{\slashed \partial}-m)O^\dagger(\vec{p},\vec{x})&=-\gamma^0\vec{\slashed
  p}O^\dagger(\vec{p},\vec{x})i\gamma^0-\gamma^0\slashed p O^\dagger(\vec{p},\vec{x})i\gamma^0\\
&=-O^\dagger(\vec{p},\vec{x})i\gamma^0E_p
\end{align}
From remark \ref{rem:derivative}, the operator
$i\gamma^0(i\vec{\slashed \partial}-m)$ is skew-hermitian and so:

\begin{align}
&\Big(\int d^3\vec{x}
O(\vec{p},\vec{x})i\gamma^0(i\vec{\slashed \partial}-m)O^\dagger(\vec{q},\vec{x})\Big)^\dagger=-\int d^3\vec{x}
O(\vec{q},\vec{x})i\gamma^0(i\vec{\slashed \partial}-m)O^\dagger(\vec{p},\vec{x})
\end{align}
Which implies:
\begin{align}
\int d^3\vec{x}
i\gamma^0E_qO(\vec{q},\vec{x})O^\dagger(\vec{p},\vec{x})&=\int d^3\vec{x}
O(\vec{q},\vec{x})O^\dagger(\vec{p},\vec{x})i\gamma^0E_p
\end{align}
Noting that $E_p+E_q>0$, this implies that:
\begin{align}
&\int d^3\vec{x}e^{-i\gamma^0\vec{q} \cdot
\vec{x}}\frac{\vec{\slashed q} \gamma^0(E_p+m)+\vec{\slashed p} \gamma^0(E_q+m)}{\sqrt{E_q+m}\sqrt{2E_q}\sqrt{E_p+m}\sqrt{2E_p}}e^{i\gamma^0\vec{p} \cdot
\vec{x}}=0
\end{align}
Therefore, we get:
\begin{align}
&\int d^3\vec{x}
O(\vec{q},\vec{x})O^\dagger(\vec{p},\vec{x})=
\int d^3\vec{x}e^{-i\gamma^0\vec{q} \cdot
\vec{x}}\frac{(E_p+m)(E_q+m)+\vec{\slashed q} \gamma^0\vec{\slashed p} \gamma^0}{\sqrt{E_q+m}\sqrt{2E_q}\sqrt{E_p+m}\sqrt{2E_p}}e^{i\gamma^0\vec{p} \cdot
\vec{x}}\\
&=\int d^3\vec{x}e^{-i\gamma^0(\vec{q}-\vec{p})\cdot
\vec{x}}\frac{(E_p+m)(E_q+m)+\vec{\slashed q} \gamma^0\vec{\slashed p}
\gamma^0}{\sqrt{E_q+m}\sqrt{2E_q}\sqrt{E_p+m}\sqrt{2E_p}}\\
&=(2\pi)^3\delta^3(\vec{q}-\vec{p})\frac{(E_p+m)(E_p+m)+\vec{p}^2}{(E_p+m)2E_p}\\
&=(2\pi)^3\delta^3(\vec{q}-\vec{p})\frac{(E_p+m)(E_p+m)+(E_p+m)(E_p-m)}{(E_p+m)2E_p}\\
&=(2\pi)^3\delta^3(\vec{q}-\vec{p})
\end{align}

The other way around:
\begin{align}
&\int
\frac{d^3\vec{p}}{(2\pi)^3}O^{\dagger}(\vec{p},\vec{y})O(\vec{p},\vec{x})=
\int \frac{d^3\vec{p}}{(2\pi)^3}\frac{\slashed p \gamma^0+m}{\sqrt{E_p+m}\sqrt{2E_p}}e^{i\gamma^0\vec{p} \cdot
(\vec{y}-\vec{x})}\frac{\slashed p
\gamma^0+m}{\sqrt{E_p+m}\sqrt{2E_p}}\\
&=\int
\frac{d^3\vec{p}}{(2\pi)^3} e^{i\frac{\slashed p}{m}\vec{p} \cdot
(\vec{y}-\vec{x})}\frac{\slashed p \gamma^0}{E_p}\\
&=\int \frac{d^3\vec{p}}{(2\pi)^3}\cos(\vec{p} \cdot(
\vec{y}-\vec{x}))+\\
&+\int \frac{d^3\vec{p}}{(2\pi)^3}(-\cos(\vec{p} \cdot(
\vec{y}-\vec{x}))\frac{\vec{\slashed p} \gamma^0}{E_p}+\sin(\vec{p} \cdot(
\vec{y}-\vec{x}))\frac{m i\gamma^0}{E_p}\\
&=\delta^3(\vec{y}-\vec{x})
\end{align}
Note that both $\cos(\vec{p} \cdot(
\vec{y}-\vec{x}))\frac{\vec{\slashed p} \gamma^0}{E_p}$ and $\sin(\vec{p} \cdot(
\vec{y}-\vec{x}))\frac{mi\gamma^0}{E_p}$ are odd in $\vec{p}$ and
therefore do not contribute to the integral.
\end{proof}

\section{Angular momentum of Majorana spinors}
\subsection{Majorana Spin}
\begin{defn}
The Majorana spin operators $\frac{1}{2}\sigma^k$ are defined as:
\begin{align}
\frac{1}{2}\sigma^k\equiv&\frac{1}{2}\gamma^k\gamma^5 ,\  \ k=1,2,3
\end{align}
\end{defn}

They verify the angular momentum algebra:
\begin{align}
[\frac{1}{2}\sigma^i,\frac{1}{2}\sigma^j]=&i\gamma^0\epsilon^{ijk}\frac{1}{2}\sigma^k
\end{align}
Where $\epsilon^{ijk}$ is the Levi-Civita symbol. Note that $i\gamma^0$ commutes with $\sigma^k$ and squares to $-1$,
so it plays the role of the imaginary unit in the angular momentum
algebra.

The eigenstates of $\frac{1}{2}\sigma^3$ are the Majorana spinors $\psi$ verifying:
\begin{align}
\psi_\pm=\frac{1\pm \sigma^3}{2}\psi_\pm
\end{align}
The eigenvalues are $\frac{1}{2}\sigma^3\psi_\pm=\pm\frac{1}{2}\psi_\pm$.

\subsection{Majorana orbital angular momentum}
\begin{defn}
A set $S$ of elements of an Hilbert space $H$ with internal product
$<,>$, is an orthonormal basis if:

1) For all $a\in S$: $<a,a>=1$;

2) (orthogonality) For all $a,b\in S$, with $a\neq b$:  $<a,b>=0$;

3) (completeness) For all $f,g\in H$, $<g,f>=\sum_{a\in S}<g,a><a,f>$.
\end{defn}

\begin{defn}
Let $\vec{x}\in \mathbb{R}^3$. The spherical coordinates
parametrization is:
\begin{align}
\vec{x}=r(\sin(\theta)\sin(\varphi)\vec{e_1}+\sin(\theta)\sin(\varphi)\vec{e_2}+\cos(\theta)\vec{e}_3)
\end{align}
where $\{\vec{e}_1,\vec{e}_2,\vec{e}_3\}$ is an orthonormal basis of
$\mathbb{R}^3$ and $r\in [0,+\infty[$ $\theta \in [0,\pi]$, $\varphi
\in [-\pi,\pi]$.
\end{defn}

\begin{defn}
$L^2(S^2)$ is the Hilbert space of real functions with domain $S^2\equiv
\{\vec{x}\in \mathbb{R}^3:|\vec{x}|=1\}$,
whose square is Lebesgue integrable in $S^2$.
The internal product is:
\begin{align}
<f,g>\equiv\int d(\cos\theta)d\varphi f(\theta,\varphi)g(\theta,\varphi),\ f,g \in L^2(S^2)
\end{align} 
\end{defn}

\begin{defn}
$L^2_4(S^2)$ is the Hilbert space of Majorana spinors whose
4 real components in the Majorana bases are in $L^2(S^2)$.
The internal product is:
\begin{align}
<\Psi,\Phi>\equiv\int d(\cos\theta)d\varphi\ \Psi^\dagger(\theta,\varphi)\Phi(\theta,\varphi),\ \Psi,\Phi \in L^2_4(S^2)
\end{align} 
\end{defn}

\begin{defn}
The Majorana angular momentum operators $\vec{L}_k$ are:
\begin{align}
\vec{L}_k\equiv \sum_{i,j=1,2,3}-i\gamma^0\epsilon_{ijk}x^i\partial_j,\ k=1,2,3
\end{align}
Where $\epsilon_{ijk}$ is the Levi-Civita symbol.
\end{defn}

The operators verify the angular momentum algebra:
\begin{align}
[\vec{L}_i,\vec{L}_j]=&i\gamma^0\epsilon_{ijk}\vec{L}_k
\end{align}
In spherical coordinates:
\begin{align}
i\gamma^0\vec{L}_3&=\partial_\varphi\\
(\vec{L})^2&=-\sin(\theta)\partial_\theta\Big(\sin(\theta)\partial(\theta)\Big)-\frac{1}{\sin^2(\theta)}\partial^2_\varphi
\end{align}
\begin{defn}
The cosine spherical harmonics $Y^c_{lm}$, sine spherical harmonics
$Y^s_{lm}$ and associated Legendre functions of the first kind
$P_{lm}$ are:
\begin{align}
Y^c_{lm}(\theta,\varphi)&\equiv\sqrt{\frac{2l+1}{4\pi}\frac{(l-m)!}{(l+m)!}}
P_{l}^m(\cos\theta)\cos(m \varphi)\\
Y^s_{lm}(\theta,\varphi)&\equiv\sqrt{\frac{2l+1}{4\pi}\frac{(l-m)!}{(l+m)!}}
P_{l}^m(\cos\theta)\sin(m \varphi)\\
P_{l}^m(\xi)&\equiv\frac{(-1)^{m}}{2^{l}l!}(1-\xi^{2})^{m/2}
\frac{\mathrm{d}^{l+m}}{\mathrm{d}\xi^{l+m}}(\xi^{2}-1)^{l}
\end{align}
where $\theta \in [0,\pi]$, $\varphi
\in [-\pi,\pi]$, $\xi\in [-1,1]$ and $l,m$ are integer numbers
$l\geq 0$, $-l\leq m\leq l$.
\end{defn}

The spherical harmonics verify \cite{harmonics}:
\begin{align}
\partial_\varphi Y^c_{lm}(\theta,\varphi)&=-m
Y^s_{lm}(\theta,\varphi)\\
\partial_\varphi Y^s_{lm}(\theta,\varphi)&=m
Y^c_{lm}(\theta,\varphi)\\
-\Big(\sin(\theta)\partial_\theta\Big(\sin(\theta)\partial_\theta\Big)
+\frac{1}{\sin^2(\theta)}\partial^2_\varphi\Big)Y^a_{lm}&=l(l+1)Y^a_{lm},\ a=c,s
\end{align}
\begin{rem}
\label{rem:harmonics}
The spherical harmonics verify:
\begin{align}  
<Y^{s}_{l'm'},Y^{c}_{lm}>&=0\\
<Y^{s}_{l'm'},Y^{s}_{lm}>+<Y^{c}_{l'm'},Y^{c}_{lm}>&=\delta_{l'l}\delta_{m'm}
\end{align}
For all $f,g\in L^2(S^2)$:
\begin{align}
<g,f>=\sum_{a=c,s,\ l\geq 0,\ -l\leq m\leq
  l} <g,Y_{lm}^a><Y_{lm}^a,f>
\end{align}
\end{rem}

\begin{defn}
The Majorana spherical harmonics $Y_{lm}$ are:
\begin{align}
Y_{lm}(\theta,\varphi)&\equiv Y^c_{lm}(\theta,\varphi)+i\gamma^0Y^s_{lm}(\theta,\varphi)\\
&=\sqrt{\frac{2l+1}{4\pi}\frac{(l-m)!}{(l+m)!}}
P_{l}^m(\cos\theta)e^{i\gamma^0 m \varphi}
\end{align}
\end{defn}

The Majorana spherical harmonics are similar to the standard Laplace spherical
harmonics definition, with $i\gamma^0$ in place of $i$. The properties
are also similar.

They verify:
\begin{align}
(\vec{L}_3-m) Y_{lm}(\vec{x})&=0\\
(\vec{L}^2-l(l+1)) Y_{lm}(\vec{x})&=0
\end{align}

\begin{prop}
The columns of the Majorana spherical harmonics matrices form an 
orthonormal basis of the Hilbert space $L^2_4(S^2)$.
\end{prop}

\begin{proof}
We apply the remark \ref{rem:harmonics} to directly obtain:
\begin{align}  
\int d(\cos\theta)d\varphi
Y^\dagger_{l'm'}(\theta,\varphi)Y_{lm}(\theta,\varphi)&=\delta_{l'l}\delta_{m'm}
\end{align}
For all $\Phi,\Psi \in L^2_4(S^2)$:
\begin{align}
<\Phi,\Psi>&=\sum_{l\geq 0,\ -l\leq m\leq
  l} <\Phi,Y_{lm}\psi_{lm}>\\
\psi_{lm}&\equiv \int d(\cos\theta)d\varphi
Y^\dagger_{lm}(\theta,\varphi)\Psi(\theta,\varphi)
\end{align}
\end{proof}

\subsection{Majorana total angular momentum space}
The operator $\vec{\sigma}\cdot \vec{L}$ is:
\begin{align}
\vec{\sigma}\cdot
\vec{L}&=-i\gamma^0\epsilon^{ij}_k\sigma^k x_i\partial_j\\
&=-\frac{[\sigma^i,\sigma^j]}{2} x_i\partial_j\\
&=\frac{\gamma^i\gamma^j-\gamma^j\gamma^i}{2}x_i \partial_j,\ i,j=1,2,3
\end{align}
In spherical coordinates:
\begin{align}
i\slashed \partial&=i\gamma^r (\partial_r-\frac{1}{r}\vec{\sigma}\cdot \vec{L})\\
\vec{\sigma}\cdot \vec{L}&=\gamma^\theta\gamma^r
\partial_\theta+\gamma^\varphi\gamma^r \frac{1}{sin
  \theta}\partial_\varphi
\end{align}
$\theta$ and $\varphi$ are the angles of $\vec{x}$ in spherical
coordinates, $r$ is the radius.

It verifies:
\begin{align}
\vec{\sigma}\cdot\vec{L}=(\vec{L}+\frac{1}{2}\vec{\sigma})^2-\vec{L}^2-\frac{3}{4}
\end{align}
The term $\vec{L}+\frac{1}{2}\vec{\sigma}$ is the sum of two angular
momentum operators of integer and one-half spin.

\begin{rem}
\label{rem:clebsh}
Let $\vec{L}$ be an integer spin angular
momentum operator, with orthonormal eigenstates $|l,m>$.
Let $\frac{1}{2}\vec{\sigma}$ be a spin one-half angular
momentum operator, with orthonormal eigenstates $|\frac{1}{2},s>$, where $s=\pm
\frac{1}{2}$.
Then, the orthonormal eigenstates of the operator
$\vec{L}+\frac{1}{2}\vec{\sigma}$, are given by \cite{harmonics}:
\begin{align}
|j,\mu,(j+1/2)>=&-\sqrt{\frac{j-\mu+1}{2j+2}}
|j+1/2,\mu-1/2>|\frac{1}{2},+\frac{1}{2}>\\
&+\sqrt{\frac{j+\mu+1}{2j+2}} |j+1/2,\mu+1/2>|\frac{1}{2},-\frac{1}{2}>
\end{align}

\begin{align}
|j,\mu,(j-1/2)>=&+\sqrt{\frac{j+\mu}{2j}}
|j-1/2,\mu-1/2>|\frac{1}{2},+\frac{1}{2}>\\
&+\sqrt{\frac{j-\mu-1}{2j}} |j-1/2,\mu+1/2>|\frac{1}{2},-\frac{1}{2}>
\end{align}
Where $j=\frac{1}{2},\frac{3}{2},...$ and $-j\leq \mu\leq j$.
They satisfy:
\begin{align}
(\vec{L}_3+\frac{\sigma^3}{2})|j,\mu,(j\pm 1/2)>&=\mu |j,\mu,(j\pm
1/2)>\\
(\vec{L}+\frac{\vec{\sigma}}{2})^2|j,\mu,(j\pm 1/2)>&=j(j+1) |j,\mu,(j\pm 1/2)>\\
\vec{\sigma} \cdot \vec{L} |j,\mu,(j\pm 1/2)>&=-(\pm(j+1/2)+2)|j,\mu,(j\pm 1/2)>\\
\sigma^r\ |j,\mu,(j+ 1/2)>&=-|j,\mu,(j-1/2)>
\end{align}
\end{rem}

\begin{defn}
The Majorana spherical matrices are:
\begin{align}
\Omega_{l \mu}(\theta,\varphi)
&\equiv \Big(-\sqrt{\frac{l-\mu}{2l+1}}
Y_{l,\mu}(\theta,\varphi)+\sqrt{\frac{l+\mu+1}{2l+1}}
Y_{l,\mu+1}(\theta,\varphi)\sigma^1\Big)\frac{1+\sigma^3}{2}\\
&+\Big(\sqrt{\frac{l+\mu}{2l-1}}
Y_{l-1,\mu}(\theta,\varphi)\sigma^1+\sqrt{\frac{l-\mu-1}{2l-1}}
Y_{l-1,\mu+1}(\theta,\varphi)\Big)\frac{1-\sigma^3}{2}
\end{align}
with the integers $l\geq 1$ and $-l\leq \mu \leq l-1$. $Y_{l\mu}$ the Majorana spherical harmonics.
\end{defn}

\begin{prop}
The columns of the Majorana spherical harmonics matrices form a
complete orthonormal basis of the Hilbert space $L^2_4(S^2)$.
\end{prop}

\begin{proof}
Using remark \ref{rem:clebsh}, after some calculations, we get:
\begin{align}  
\int d(\cos\theta)d\varphi \Omega^\dagger_{l'\mu'}(\theta,\varphi)\Omega_{l\mu}(\theta,\varphi)&=\delta_{l'l}\delta_{\mu'\mu}\\
\sum_{l\geq 1,\ -l\leq \mu\leq l-1}\int d(\cos\theta)d\varphi\ \Phi^\dagger(\theta,\varphi) \Omega_{l\mu}(\theta,\varphi)\psi_{l\mu}&=\int d(\cos\theta)d\varphi\ \Phi^\dagger(\theta,\varphi) \Psi(\theta,\varphi)
\end{align}
For all $\Phi \in L^2_4(S^2)$.
\end{proof}

Using remark \ref{rem:clebsh}, the Majorana spherical matrices verify:
\begin{align}
(\vec{L}^3+\frac{\sigma^3}{2})\
\Omega_{l \mu}&=(\mu+\frac{1}{2})\Omega_{l \mu}\\
\vec{\sigma} \cdot \vec{L}\
\Omega_{l \mu}&=-\Omega_{l \mu}(l \sigma^3+1)\\
\sigma^r\ \Omega_{l \mu}&=-\Omega_{l \mu}\sigma^1\\
i\gamma^r\Omega_{l \mu}&=(-1)^\mu\Omega_{l,-\mu-1}i\gamma^5\\
\vec{\sigma} \cdot \vec{L}i\gamma^r\
\Omega_{l \mu}&=i\gamma^r\Omega_{l \mu}(l\sigma^3-1)
\end{align}

\subsection{Radial Momentum Space}
\begin{rem}
The spherical Bessel functions of the first kind, $j_l:\mathbb{R}_+\to
\mathbb{R}$ with the
integer $l\geq 0$, verify:
\begin{align}
&(\partial^2_r+\frac{2}{r}\partial_r-\frac{l(l+1)}{r^2})j_l(pr)=-p^2j_l(pr)\\
&\int_0^{+\infty}dr\ r^2
j_l(pr)j_l(p'r)=\frac{\pi\delta(p-p')}{2 p^2}\\
&\int_0^{+\infty}\frac{dp\ 2p^2}{\pi} j_l(pr)j_l(p r')=\frac{\delta(r-r')}{r^2}
\end{align}
Where the Dirac delta $\delta$ is such that for all $f\in L^2(\mathbb{R})$:
\begin{align}
f(0)&=\int dx\ \delta(x)f(x)
\end{align}
\end{rem}

\begin{defn}
The Hankel-Majorana Transform $\psi(p,l,\mu)$ of a Majorana spinor
$\Psi(\vec{x})\in L^2_4(\mathbb{R}^3)$ is the Majorana spinor:
\begin{align}
\psi(p,l,\mu)&\equiv \int dr d(cos\theta)d\varphi r^2
\Lambda^\dagger(p,l,\mu,r,\theta,\varphi)\Psi(r,\theta,\varphi)\\
\Lambda(p,l,\mu,r,\theta,\varphi)&\equiv \Big(p
j_l(pr)+(E_p-m)j_{l-1}(pr)i\gamma^r\Big)
\Omega_{l \mu}(\theta,\varphi)\frac{1+\sigma^3}{2}\\
&+\Big(p j_{l-1}(pr)-(E_p-m) j_{l}(pr)i\gamma^r\Big)\Omega_{l \mu}(\theta,\varphi)\frac{1-\sigma^3}{2}
\end{align}
Where $\Lambda$ are the Hankel-Majorana matrices, $m,p\geq 0$,
$E_p=\sqrt{p^2+m^2}$ and the integers $l\geq 1, -l\leq \mu\leq l$.
\end{defn}

\begin{prop}
\label{prop:Hankel}
Let $\psi(p,l,\mu)$ be the Hankel-Majorana Transform of a Majorana spinor
$\Psi\in L^2_4(\mathbb{R}^3)$. The inverse Hankel-Majorana Transform of
$\psi(p,l,\mu)$ is:
\begin{align}
\Psi'(r,\theta,\varphi)\equiv \sum_{l\geq 1, -l\leq \mu\leq l}\int_0^{+\infty} \frac{dp\ (E_p+m)}{E_p\pi}\Lambda(p,l,\mu,r,\theta,\varphi)\psi(p,l,\mu)
\end{align}
It verifies, for all $\Phi\in L^2_4(\mathbb{R}^3)$:
\begin{align}
\int d(cos\theta)d\varphi\ dr\ r^2\
\Phi^\dagger(r,\theta,\varphi)\Psi'(r,\theta,\varphi)= 
\int d(cos\theta)d\varphi\ dr\ r^2\Phi^\dagger(r,\theta,\varphi)\Psi(r,\theta,\varphi)
\end{align}
\end{prop}

\begin{proof}
The following equation is verified:
\begin{align}
i\gamma^0(i\vec{\slashed \partial}-m)\Lambda(p,l,\mu)=E_p\Lambda(p,l,\mu)i\gamma^0
\end{align}
Since the operator $i\gamma^0(i\vec{\slashed \partial}-m)$ is skew-Hermitic the equation above implies that:
\begin{align}
&i\gamma^0 E_{p'}I=I i\gamma^0 E_{p}\\
&I\equiv \int d(cos\theta)d\varphi\ dr\ r^2\ 
\Lambda^\dagger(p',l',\mu',r,\theta,\varphi)\Lambda(p,l,\mu,r,\theta,\varphi)
\end{align}
As $E_p+E_{p'}>0$, in the integral $I$ the terms odd in $i\gamma^r$
are null. From the orthogonality of the spherical matrices, we get
that the $\Lambda$ matrices are orthogonal:
\begin{align}
I&=\delta_{l'l}\delta_{\mu'\mu}\int d(cos\theta)d\varphi\ dr\ r^2\\
&\Big( p' j_l(p'r)p
j_l(pr)+(E_{p'}-m)j_{l-1}(p'r)(E_p-m)j_{l-1}(pr)\frac{1+\sigma^3}{2}\\
&+p' j_{l-1}(p'r)p
j_{l-1}(pr)+(E_{p'}-m)j_{l}(p'r)(E_p-m)j_{l}(pr)\frac{1-\sigma^3}{2}\Big)\\
&=\delta_{l'l}\delta_{\mu'\mu}\frac{\pi\delta(p-p')}{2
  p^2}(E_p-m)2E_p=
\delta_{l'l}\delta_{\mu'\mu}\frac{\pi E_p\delta(p-p')}{E_p+m}
\end{align}
To show completeness, using $i\gamma^r\Omega_{l
  \mu}=(-1)^\mu\Omega_{l,-\mu-1}i\gamma^5$, we first show that:
\begin{align}
&\sum_{l'\mu'} \int d(cos\theta)d\varphi\
\psi^\dagger(p,l',\mu')\Lambda^\dagger(p,l',\mu',r,\theta,\varphi)\Omega_{l\mu}(\theta,\varphi)=\\
&=\psi^\dagger(p,l,\mu)p
(j_l(pr)\frac{1+\sigma^3}{2}+j_{l-1}(pr)\frac{1-\sigma^3}{2})\\
&+\psi^\dagger(p,l,-\mu-1)(-1)^\mu(E_p-m)
(-j_l(pr)\frac{1-\sigma^3}{2}+j_{l-1}(pr)\frac{1+\sigma^3}{2})i\gamma^5
\end{align}
\begin{align}
&=\int d(cos\theta')d\varphi' dr' (r')^2\
\Psi^\dagger(r',\theta',\varphi')\Big(\\
& pj_l(pr')\Big(p j_l(pr)+(E_p-m)j_{l-1}(pr)i\gamma^r\Big)\Omega_{l
  \mu}(\theta',\varphi')\frac{1+\sigma^3}{2}\\
&+pj_{l-1}(pr')\Big(p j_{l-1}(pr)-(E_p-m)
j_{l}(pr)i\gamma^r\Big)\Omega_{l \mu}(\theta',\varphi')\frac{1-\sigma^3}{2}\\
&(-1)^\mu(E_p-m)j_{l-1}(pr')\Big(p j_l(pr)+(E_p-m)j_{l-1}(pr)i\gamma^r\Big)\Omega_{l,
  -\mu-1}(\theta',\varphi')\frac{1+\sigma^3}{2}i\gamma^5\\
&-(-1)^\mu(E_p-m)j_{l}(pr')\Big(p j_{l-1}(pr)-(E_p-m)
j_{l}(pr)i\gamma^r\Big)\Omega_{l, -\mu-1}(\theta',\varphi')\frac{1-\sigma^3}{2}i\gamma^5
\end{align}

\begin{align}
&=\int d(cos\theta')d\varphi' dr' (r')^2\
\Psi^\dagger(r',\theta',\varphi')\Big(\\
& pj_l(pr')\Big(p j_l(pr)+(E_p-m)j_{l-1}(pr)i\gamma^r\Big)\Omega_{l
  \mu}(\theta',\varphi')\frac{1+\sigma^3}{2}\\
&+pj_{l-1}(pr')\Big(p j_{l-1}(pr)-(E_p-m)
j_{l}(pr)i\gamma^r\Big)\Omega_{l \mu}(\theta',\varphi')\frac{1-\sigma^3}{2}\\
&(E_p-m)j_{l-1}(pr')\Big(p j_l(pr)+(E_p-m)j_{l-1}(pr)i\gamma^r\Big)\Omega_{l,\mu}(\theta',\varphi')\frac{1-\sigma^3}{2}\\
&-(E_p-m)j_{l}(pr')\Big(p j_{l-1}(pr)-(E_p-m)
j_{l}(pr)i\gamma^r\Big)\Omega_{l, \mu}(\theta',\varphi')\frac{1+\sigma^3}{2}
\end{align}
\begin{align}
&=\int d(cos\theta')d\varphi' dr' (r')^2\
\Psi^\dagger(r',\theta',\varphi')\Omega_{l
  \mu}\frac{2p^2E_p}{E_p+m}\Big(\\
&j_l(pr')j_l(pr)\frac{1+\sigma^3}{2}+j_{l-1}(pr')j_{l-1}(pr)\frac{1-\sigma^3}{2}\Big)
\end{align}
If we integrate on $p$ and use the completeness of the spherical
Bessel functions, we get:
\begin{align}
\int d(cos\theta)d\varphi
\Psi^{'\dagger}(r,\theta,\varphi)\Omega_{l\mu}(\theta,\varphi)=
\int d(cos\theta)d\varphi \Psi^{\dagger}(r,\theta,\varphi)\Omega_{l\mu}(\theta,\varphi)
\end{align}
Since the columns of the spherical matrices $\Omega_{l\mu}$ are a complete basis, we
have shown the completeness of the Hankel-Majorana transform:
\begin{align}
\int d(cos\theta)d\varphi dr\ r^2\ 
\Psi^{'\dagger}(r,\theta,\varphi)\Phi(r,\theta,\varphi)=
\int d(cos\theta)d\varphi dr\ r^2 \Psi^{\dagger}(r,\theta,\varphi)\Phi(r,\theta,\varphi)
\end{align}
For all $\Phi \in L^2_4(\mathbb{R}^3)$.
\end{proof}

\section{Relation between the Dirac and Majorana Momentums}
The Dirac equation for the free fermion can be written as:
\begin{equation}
i\gamma^0(i\slashed \partial-m)\Psi(x)=0
\end{equation}
Where $\Psi$ is a spinor. Note that the equation contains only
Majorana matrices. The Fourier or Hankel Transforms of the equation are:
\begin{align}
(\partial_0+i\gamma^0E_p)\Psi(x^0,p)=0
\end{align}
The solutions can be written as:
\begin{align}
\Psi(x)=\int \frac{d^3\vec{p}}{(2\pi)^3} \frac{\slashed p
  \gamma^0+m}{\sqrt{E_p+m}\sqrt{2E_p}}e^{-i\gamma^0 p \cdot x}\psi(\vec{p})
\end{align}
Where $p^0=E_p$ and $\psi(\vec{p})$ is an arbitrary spinor. If $\psi(\vec{p})$ is
a Majorana spinor, then the solution $\Psi(x)$ is also a Majorana
spinor.

The solutions can also be written as:
\begin{align}
\Psi(x^0, r,\theta,\varphi)= \sum_{l\geq 1, -l\leq \mu\leq
  l-1}\int_0^{+\infty} \frac{dp(E_p+m)}{E_p\pi}\Lambda(p,l,\mu,r,\theta,\varphi)e^{-i\gamma^0 E_p \cdot
  x^0} \psi(p,l,\mu)
\end{align}
Where $\psi(p,l,\mu)$ is an arbitrary spinor and $\Lambda$ are the
Hankel-Majorana matrices.

The set of quantum numbers $(\vec{p})$ and $(p,l,\mu)$ are related
with the linear and spherical momentums of Dirac spinors.
The Majorana spin is related with the Dirac spin. 
For instance, to obtain the Dirac spinor solution for the free electron, we just set
$\psi_e(\vec{p})=\frac{1+\gamma^0}{2}\psi_e(\vec{p})$ and we get:
\begin{align}
\Psi_e(x)=\int \frac{d^3\vec{p}}{(2\pi)^3} \frac{\slashed p+m}{\sqrt{E_p+m}\sqrt{2E_p}}e^{-i p \cdot x}\frac{1+\gamma^0}{2}\psi_e(\vec{p})
\end{align}
The matrix $\gamma^0$ was replaced by the identity matrix $1$, due to the
presence of the projector. The same thing happens with the spherical
solution and with the spin.

To obtain the Dirac spinor solution for the free positron, we just set
$\psi_p(\vec{p})=\frac{1-\gamma^0}{2}\psi_p(\vec{p})$ and the matrix
$\gamma^0$ gets replaced by $-1$.

\section{Energy-momentum space}
Now we can extend our transforms to define an energy-momentum
space. We will use the notation:
\begin{align}
[\slashed p]&=\gamma^0 E_p-\vec{\gamma}\cdot\vec{p}
\end{align}
Note that $\slashed p$ is not necessarily on-shell, while $[\slashed
p]$ is on-shell, that is $([\slashed p])^2=m^2$. 
Both $E_p$ and $[\slashed p]$  do not depend on $p^0$.

\begin{defn}
Given a Majorana spinor $\Psi \in L^2_4(\mathbb{R}^4)$, the Fourier-Majorana transform
in space-time is defined as:
\begin{align}
\psi(p)&\equiv \int d^4x O(p,x)\Psi(x)
\end{align}
Where $O(p,x)$ is:
\begin{align}
O(p,x)&\equiv e^{i\gamma^0 p^0 x^0}O(\vec{p},\vec{x})=e^{i\gamma^0 p \cdot
  x}\frac{[\slashed p] \gamma^0+m}{\sqrt{E_p+m}\sqrt{2E_p}}
\end{align}
Note that $E_p$ and $[\slashed p]=\gamma^0 E_p-\vec{\gamma}\cdot\vec{p}$ don't depend on $p^0$, but
$p\cdot x=p^0x^0-\vec{p}\cdot\vec{x}$ does.
\end{defn}

\begin{prop}
The inverse Fourier-Majorana transform in space-time is given by:
\begin{align}
\Psi(x)&=\int \frac{d^4p}{(2\pi)^4}O^\dagger(p,x)\psi(p)
\end{align}
Where $O^{\dagger}$ is the hermitian conjugate of $O$, given by:
\begin{align}
O^\dagger(p,x)=O^\dagger(\vec{p},\vec{x})e^{-i\gamma^0 p^0 \cdot
  x^0}=\frac{[\slashed p]
  \gamma^0+m}{\sqrt{E_p+m}\sqrt{2E_p}}e^{-i\gamma^0 p \cdot x}
\end{align}
\end{prop}

\begin{proof}
\begin{align}
\int
\frac{d^4p}{(2\pi)^4}O^{\dagger}(p,y)O(p,x)&=\int
\frac{d^3\vec{p}}{(2\pi)^3}O^{\dagger}(\vec{p},\vec{y})\Big(\int
\frac{dp^0}{2\pi} e^{-i\gamma^0p^0(y^0-x^0)}\Big)O(\vec{p},\vec{x})\\
&=\delta(y^0-x^0)\int
\frac{d^3\vec{p}}{(2\pi)^3}O^{\dagger}(\vec{p},\vec{y})O(\vec{p},\vec{x})\\
&=\delta^4(y-x)
\end{align}
\begin{align}
\int d^4x O(q,x)O^{\dagger}(p,x)&=\int dx^0 e^{i\gamma^0 q^0 x^0}\Big(\int
d^3\vec{x}O(\vec{q},\vec{x})O^{\dagger}(\vec{p},\vec{x})\Big)e^{-i\gamma^0 p^0
  x^0}\\
&=(2\pi)^3\delta^3(\vec{q}-\vec{p})\int dx^0 e^{i\gamma^0
  (q^0-p^0)x^0}\\
&=(2\pi)^4\delta^4(q-p)
\end{align}
\end{proof}

\begin{defn}
The Hankel-Majorana transform in space-time of a Majorana spinor $\Psi
\in L^2_4(\mathbb{R}^4)$ is:
\begin{align}
\psi'(p^0,p,l,\mu)&\equiv \int dx^0 e^{i\gamma^0 p^0 x^0} \psi(x^0,p,l,\mu)
\end{align}
Where $\psi(x^0,p,l,\mu)$ is the Hankel-Majorana transform in space of $\Psi$.
\end{defn}

\begin{prop}
Let $\psi(p^0,p,l,\mu)$ be the Hankel-Majorana Transform in space-time
of a Majorana spinor $\Psi\in L^2_4(\mathbb{R}^4)$. The inverse Hankel-Majorana Transform of
$\psi(p^0,p,l,\mu)$ is:
\begin{align}
\Psi'(x^0, r,\theta,\varphi)\equiv \sum_{l\geq 1, -l\leq \mu\leq
  l}\int_0^{+\infty} \frac{dp(E_p+m)}{E_p\pi}\int_{-\infty}^{+\infty}\frac{dp^0}{2\pi}\Lambda(p,l,\mu,r,\theta,\varphi)e^{-i\gamma^0 p^0 \cdot
  x^0} \psi(p^0,p,l,\mu)
\end{align}
It verifies, for all $\Phi\in L^2_4(\mathbb{R}^4)$:
\begin{align}
&\int dx^0 d(cos\theta)d\varphi\ dr\ r^2\
\Phi^\dagger(x^0,r,\theta,\varphi)\Psi'(x^0,r,\theta,\varphi)=\\ 
&=\int dx^0 d(cos\theta)d\varphi\ dr\ r^2\Phi^\dagger(x^0,r,\theta,\varphi)\Psi(x^0,r,\theta,\varphi)
\end{align}
\end{prop}

\begin{proof}
The matrices $\Lambda(p,l,\mu,r,\theta,\varphi)e^{-i\gamma^0 p^0 \cdot
  x^0}$ are orthogonal:
\begin{align}
&\int dx^0 d(cos\theta)d\varphi\ dr\ r^2\ 
e^{i\gamma^0 p'^{0} \cdot
  x^0}\Lambda^\dagger(p',l',\mu',r,\theta,\varphi)\Lambda(p,l,\mu,r,\theta,\varphi)e^{-i\gamma^0 p^0 \cdot
  x^0}=\\
&=\delta_{l'l}\delta_{\mu'\mu}\frac{\pi E_p\delta(p-p')}{E_p+m}\int dx^0e^{i\gamma^0 p'^{0} \cdot
  x^0}e^{-i\gamma^0 p^0 \cdot
  x^0}=\delta_{l'l}\delta_{\mu'\mu}\frac{\pi E_p\delta(p-p')}{E_p+m}2\pi\delta(p'^{0}-p^0)
\end{align}

To show completeness, we first show that:
\begin{align}
&\sum_{l'\mu'}\int_0^{+\infty}\frac{dp'\ (E_p'+m)}{E_p'\pi}\int d(cos\theta)d\varphi dr\ r^2\\
&\psi^\dagger(p^0,p',l',\mu')e^{i\gamma^0 p^{0} \cdot
  x^0}\Lambda^\dagger(p',l',\mu',r,\theta,\varphi)\Lambda(p,l,\mu,r,\theta,\varphi)=\\
&=\psi^\dagger(p^0,p,l,\mu)e^{i\gamma^0 p^{0} \cdot
  x^0}\\
&=\int dx'^0 d(cos\theta)d\varphi dr\ r^2 \Psi^\dagger(x'^0,r,\theta,\varphi)\Lambda(p,l,\mu,r,\theta,\varphi) e^{-i\gamma^0 p^0 x'^0}e^{i\gamma^0 p^{0} \cdot
  x^0}
\end{align}
If we integrate on $p^0$, we get:
\begin{align}
\int d(cos\theta)d\varphi dr\ r^2 
\Psi^{'\dagger}(x^0,r,\theta,\varphi)\Lambda(p,l,\mu,r,\theta,\varphi)=
\int d(cos\theta)d\varphi dr\ r^2 \Psi^{\dagger}(r,\theta,\varphi)\Lambda(p,l,\mu,r,\theta,\varphi)
\end{align}
Since the columns of the Hankel matrices
$\Lambda(p,l,\mu,r,\theta,\varphi)$ are a complete basis, we
have shown the completeness of the Hankel-Majorana transform in space-time:
\begin{align}
&\int dx^0 d(cos\theta)d\varphi dr\ r^2\ 
\Psi^{'\dagger}(x^0,r,\theta,\varphi)\Phi(x^0,r,\theta,\varphi)=\\
&=\int dx^0 d(cos\theta)d\varphi dr\ r^2 \Psi^{\dagger}(x^0,r,\theta,\varphi)\Phi(x^0,r,\theta,\varphi)
\end{align}
For all $\Phi \in L^2_4(\mathbb{R}^4)$.
\end{proof}

\section{Conclusion}

We fulfilled our goal to show that (without second quantization operators)
all the kinematic properties of a free spin 1/2
particle with mass are present in the real solutions of the
real free Dirac equation.

Since we live in a world where the Lorentz symmetries are important,
we hope that the Majorana transforms can have some applications. 
I personally think that the study of the Majorana spinor properties
will be useful in our understanding of the Standard Model. 
In particular, since the Majorana spinors are an irreducible
representation of the double cover of the proper orthochronous Lorentz
group, like the Weyl spinor, as well as the full Lorentz group, unlike the Weyl spinor,
I think that their study might improve our knowledge about the
discrete symmetries of the Lorentz group and the interactions that
violate them.

\section*{Acknowledgements}
The work of Leonardo Pedro was supported by FCT under contract
SFRH/BD/70688/2010.

\addcontentsline{toc}{section}{References}
\bibliography{qed}{}
\bibliographystyle{utphys}

\end{document}